\documentclass{llncs}
\usepackage[boxed,lined]{algorithm2e}

\usepackage{amsfonts}
\usepackage{amsmath}
\usepackage{amssymb}
\usepackage{bbm}

\newcommand{\citep}[1]{\cite{#1}}

\newcommand{\R}{\mathbb{R}}
\renewcommand{\P}{\mathbb{P}}

\newcommand{\one}{\bbbone}

\newcommand{\Qr}{Q^{\mathrm{r}}}
\newcommand{\Qt}{Q^{\mathrm{t}}}
\newcommand{\tQr}{\tilde{Q}^{\mathrm{r}}}
\newcommand{\tQt}{\tilde{Q}^{\mathrm{t}}}
\newcommand{\mat}[1]{\mathbf{#1}}
\newcommand{\sstar}{\Sigma^\star}
\newcommand{\lang}[1]{\mathcal{L}_{#1}}
\newcommand{\pwset}[1]{\mathcal{P}(#1)}
\spnewtheorem{fact}{Fact}{\bfseries}{\itshape}

\title{Ergodicity of Random Walks on Random DFA}
\author{Borja Balle}
\institute{%
\email{bballe@cs.mcgill.ca}\\
Reasoning and Learning Laboratory\\
School of Computer Science\\
McGill University}

\begin{document}

\maketitle

\begin{abstract}
Given a DFA we consider the random walk that starts at the initial state and
at each time step moves to a new state by taking a random transition from
the current state.
This paper shows that for \emph{typical DFA} this random walk induces an
\emph{ergodic} Markov chain.
The notion of typical DFA is formalized by showing that ergodicity holds with
high probability when a DFA is sampled uniformly at random from the set of all
automata with a fixed number of states.
We also show the same result applies to DFA obtained by minimizing typical DFA.
\end{abstract}

\section{Introduction}

Deterministic finite automata (DFA) is a well-known computational model which
has been used in computer science for a long time.
In the context of learning theory, DFA's ability to succintly represent regular
languages makes them an interesting hypothesis class for learning regular
languages and other simpler concepts.
Unfortunately, it is known that learning DFA under quite general learning models
is a formidable problem \citep{angluinmq,pittwarmuth,kearnsvaliant}.
Empirical investigations, however, suggest that most DFA may not be as
hard to learn as these worst-case (conditional) lower bounds indicate
\citep{langrandom,abbadingo}.
These seemingly contradicting facts raise the question of whether existing
hardness results are excessively pessimistic and, in fact, typical DFA are
easy to learn.

A common approach to characterize the nature of typical objects inside a class
is to draw elements from the class uniformly at random and study which
properties hold with high probability.
For finite classes this amounts to showing that a certain property holds for all
but a negligible fraction of objects in the class.
Approaches of this sort have been recently used for showing that typical
decision trees and DNF formulas can be learned from examples drawn uniformly at
random from $\{0,1\}^n$ in polynomial time \citep{randomdt,randomdnf}.
In contrast, it was recently showed in \citep{angluinlower} that random decision
trees and random DNF formulas are hard to learn from statistical queries under
arbitrary distributions over $\{0,1\}^n$.

The results in \citep{angluinlower} also show that learning random DFA under
arbitrary distributions is hard in the statistical query model.
However, the question of whether random DFA can be learned under the uniform
distribution is a long-standing open problem on which very little progress has
been made.
The main obstruction for studying the learnability (and other properties) of
typical DFA seems to be our poor understanding of the structural regularities
exhibited by DFA constructed at random -- it is interesting to note how this
contrasts with the vast amount of information known about random
\emph{undirected} graphs \citep{bollobas2001random}.
Indeed, just very few results about the structure of random DFA are known; see
Section~\ref{sec:randomdfa} for details.

In this paper we provide some new insights about the structure of typical DFA by
studying the behavior of random walks on random DFA.
We show that with high probability, a random walk starting at the initial state
of a randomly contructed DFA and taking transitions at random induces an ergodic
Markov chain; that is, the state distribution in the random walk will converge
to a stationary distribution.
We also show that the same holds if one considers DFA obtained by minimizing
randomly generated DFA.
This is relevant to the learnability of typical DFA under the uniform
distribution because the state reached by each one of these examples corresponds
to the state reached by a particular realization of the random walk we just
described.
Hence, trying to understand the distribution over states reached by
uniformly generated examples seems to be a good starting point for understanding
how these examples can provide useful information for learning the function
computed by a DFA.
In addition, since the function computed by a DFA is invariant under
minimization, showing that ergodicity is conserved after minimizing the DFA is
also important.

The rest of the paper is structured as follows.
Section~\ref{sec:preliminaries} defines our notation and describes previous
results.
Our first result showing that random walks on random DFA are ergodic is proved
in Section~\ref{sec:aperiodic}.
Then Section~\ref{sec:minimization} shows that ergodicity of DFA is conserved
under minimization.
We conclude the paper in Section~\ref{sec:conclusion}.

\section{Preliminaries and Related Work}\label{sec:preliminaries}

Given a finite alphabet $\Sigma$ we use $\Sigma^\star$ to denote the set of all
strings over $\Sigma$. We use $\lambda$ to denote the empty string and write
$\Sigma^+ = \sstar \setminus \{\lambda\}$.
Given a predicate $P$ we use $\one_P$ to denote the indicator variable that takes
value $1$ if $P$ is true and value $0$ otherwise.
Given a set $X$ we write $\pwset{X}$ to denote the powerset of $X$ containing
all of its subsets.
For any positive integer $k$ we write $[k] = \{1,\ldots,k\}$.

\subsection{Finite Automata}

A \emph{non-deterministic finite automaton (NFA)} is a tuple $A =
\langle \Sigma,Q,q_0,\tau,\phi\rangle$, where $\Sigma$ is a finite alphabet, $Q$
is a finite set of \emph{states}, $q_0 \in Q$ is a distinguished \emph{initial
state}, $\tau : Q \times \Sigma \rightarrow \pwset{Q}$ is the \emph{transition
function}, and $\phi : Q \rightarrow \{0,1\}$ is the \emph{termination
function}.
The transition function can be inductively extended to a function $\tau : Q
\times \sstar \rightarrow \pwset{Q}$ by setting $\tau(q,\lambda) = \{q\}$ and
$\tau(q, x \sigma) = \cup_{q' \in \tau(q, x)} \tau(q', \sigma)$ for all $q \in
Q$, $x \in \sstar$, and $\sigma \in \Sigma$.
The \emph{characteristic function} of $A$ is $f_A : \sstar \rightarrow \{0,1\}$
defined as $f_A(x) = \vee_{q \in \tau(q_0,x)} \phi(q)$.
The \emph{language accepted by $A$} is the set $\lang{A} = f_A^{-1}(1) \subseteq
\sstar$.
For any $q \in Q$ we denote by $A_q = \langle \Sigma,Q,q,\tau,\phi\rangle$ the
NFA obtained by letting $q$ be the initial state of $A$.

We also define the following extension of $\tau$ over sets of states and
strings.
Given $Q' \subseteq Q$ and $X \subseteq \sstar$,
let
\begin{equation*}
\tau(Q',X) = \bigcup_{q \in Q'} \bigcup_{x \in X} \tau(q,x) \enspace.
\end{equation*}
We introduce special notations for the following choices of $X$: $\tau_1(Q') =
\tau(Q',\Sigma)$ and $\tau_\star(Q') = \tau(Q',\sstar)$.

A state $q' \in Q$ is \emph{accessible} from another state $q \in Q$ if there
exists a string $x \in \sstar$ such that $q' \in \tau(q,x)$; that is, if $q'
\in \tau_\star(q)$.
The set of states $\tau_\star(q_0)$ accessible from the initial state are called
\emph{reachable}.
If $q$ is also accessible from $q'$ then we say that $q$ and $q'$
\emph{communicate}.
Communication is an equivalence relation that induces a partition of $Q$ into
\emph{communicating classes}.
A communicating class $Q' \subseteq Q$ is \emph{closed} if $\tau_\star(Q') =
Q'$.
Because $\tau_\star(Q') \subseteq Q$ for every $Q' \subseteq Q$, each NFA must
have at least one closed communicating class.
Let $k > 1$.
A closed communicating class $Q'$ is \emph{$k$-periodic} if there exists a
partition of $Q'$ into $k$ parts $Q'_0, \ldots, Q'_{k-1}$ such that for any $0
\leq i \leq k-1$ we have $\tau_1(Q'_i) = Q'_{i+1 \mod k}$. If a closed
communicating class $Q'$ is not $k$-periodic for any $k > 1$ then we say it is
\emph{aperiodic}.
Every state belonging to a closed communicating class is called
\emph{recurrent}; the rest of states are called \emph{transient}.
We shall sometimes write $Q = \Qr \cup \Qt$ to denote the partition of $Q$ into
recurrent and transient states.
The following is a useful fact about accessibility of recurrent states.

\begin{fact}\label{fact:recurrent}
For all $q \in Q$ we have $\tau_\star(q) \cap \Qr \neq \varnothing$. In
addition, if $\Qr$ contains a single closed communicating class, then
$\tau_\star(q) \cap \Qr = \Qr$.
\end{fact}

We note that these definitions only depend on the transition structure defined
by $\tau$. In particular, they are independent of the termination function $\phi$.

Two states $q, q' \in Q$ are called \emph{undistinguishable} if $A_q$ and
$A_{q'}$ define the same language: $\lang{A_q} = \lang{A_{q'}}$.
Undistinguishability defines an equivalence relation known as
\emph{Myhill--Nerode equivalence}.
An NFA is \emph{minimal} if no pair of states are undistinguishable; that is,
if each Myhill--Nerode equivalence class contains only one state.

A \emph{deterministic finite automaton (DFA)} is a NFA $A =
\langle \Sigma,Q,q_0,\tau,\phi\rangle$ such that for every $q \in Q$ and $\sigma
\in \Sigma$ we have $|\tau(q,\sigma)| = 1$.
In this case we can -- and shall -- identify the transition function $\tau : Q
\times \Sigma \rightarrow \pwset{Q}$ with a function of type $\tau : Q \times
\Sigma \rightarrow Q$.
All definitions made for NFA are also valid for DFA.

We shall use $n$ to denote the number of states $|Q|$ and $r$ to denote the size
of the alphabet $|\Sigma|$ when convenient. In this case we may also identify
$Q$ with $[n]$ and $\Sigma$ with $[r]$.

\subsection{State-Merging Operations}

Let $A = \langle \Sigma,Q,q_0,\tau,\phi\rangle$ and $\tilde{A} = \langle
\Sigma,\tilde{Q},\tilde{q}_0,\tilde{\tau},\tilde{\phi}\rangle$ be two NFA over
the same alphabet.
We say that $\tilde{A}$ is obtained from $A$ by a \emph{merge operation} if
there exists a function $\Psi : Q \rightarrow \tilde{Q}$ satisfying the
following:
\begin{enumerate}
\item $\Psi$ is exhaustive,
\item $\Psi(q_0) = \tilde{q}_0$,
\item for every $q \in Q$ we have $\phi(q) = \tilde{\phi}(\Psi(q))$,
\item for every $q \in Q$ and $\sigma \in \Sigma$, if $q' \in \tau(q,\sigma)$
then $\Psi(q') \in \tilde{\tau}(\Psi(q),\sigma)$.
\end{enumerate}
If $\tilde{q} \in \tilde{Q}$ is such that $|\Psi^{-1}(\tilde{q})| > 1$, then we
say that $\tilde{q}$ is obtained by \emph{merging} all the states in
$\Psi^{-1}(\tilde{q})$.
A merge operation is \emph{elementary} if $|\tilde{Q}| = |Q| - 1$, implying that
only two states $q, q' \in Q$ are merged by $\Psi$. In this case the restriction
${\Psi|}_{Q \setminus \{q,q'\}}$ is a bijection onto $\tilde{Q} \setminus
\{\Psi(q)\}$.
It is immediate to verify by induction on the length of $x$ that the following
holds for any merging operation $\Psi$:
for all $q \in Q$ and $x \in \sstar$, $q' \in \tau(q,x)$ implies $\Psi(q')
\in \tilde{\tau}(\Psi(q),x)$.

\subsection{DFA Minimization}

DFA minimization is an operation that starts with a DFA $A$ recognizing a
language $\lang{}$ and yields a new DFA $A'$ with minimal size among those that
recognize $\lang{}$.
Minimization algorithms for DFA have been extensively studied in the literature,
see \citep{minimization} for a comprehensive review.
Here we describe a simple minimization algorithm based on state-merging
operations. We will use this algorithm to study how the structure of a DFA is
modified by the minimization procedure.

What follows is
a high-level description of the algorithm, making special emphasis on
the steps that actually modify the structure of the DFA. The algorithm uses a
subroutine called \texttt{MyhillNerodeClasses} to partition the set of states
$Q$ into equivalence classes of undistinguishable states. This can be done in
several ways -- e.g.\ using
Hopcroft's algorithm \citep{hopcroft} based on partition refinement -- but the
details are not relevant to us.
Given a DFA $A$ the algorithm works as follows.
First, remove all unreachable states.
Second, partition the remaining states into undistinguishable equivalence
classes.
And third, apply a sequence of elementary merge operations to collapse
each set in the partition into a single state.
This merging process will start and end with a DFA but may produce an NFA in its
intermediate steps.
Pseudocode for this minimization algorithm is given in Figure~\ref{fig:algmin}.
%

\begin{figure}
\begin{center}
\begin{algorithm}[H]
\SetKwFunction{MyhillNerodeClasses}{MyhillNerodeClasses}
\DontPrintSemicolon
\KwIn{DFA $A = \langle \Sigma,Q,q_0,\tau,\phi\rangle$}
\KwOut{Minimal DFA $\tilde{A}$}
\tcp{Remove unreachable states}
Let $\tilde{Q} \leftarrow \tau_\star(q_0)$\;
Let $\tilde{A} \leftarrow \langle \Sigma, \tilde{Q}, q_0, \tau_{|\tilde{Q}},
\phi_{|\tilde{Q}} \rangle$\;
\tcp{Partition $\tilde{Q}$ into classes of undistinguishable states}
Let $(P_1,\ldots,P_s) \leftarrow $ \MyhillNerodeClasses{$\tilde{A}$}\;
\tcp{Perform state-merging operations}
\ForEach{$i \in [s]$}{
\While{$|P_i| > 1$}{
Let $\Psi$ be an elementary merge operation merging any two $q, q' \in P_i$\;
Let $\tilde{A} \leftarrow \Psi(\tilde{A})$\;
Let $P_i \leftarrow (P_i \setminus \{q,q'\}) \cup \{\Psi(q)\}$\;
}
}
\KwRet{$\tilde{A}$}\;
\end{algorithm}
\caption{DFA minimization algorithm using state-merging
operations}\label{fig:algmin}
\end{center}
\end{figure}

\subsection{Random Walks on Finite Automata}\label{sec:randomwalk}

Given an NFA $A$, a \emph{random walk} on $A$ is a realization of the following
Markov chain over the state space $Q$: starting at the initial state $q_0$, for
$t \geq 0$ we choose a next state $q_{t+1}$ from $Q$ at random according to a
distribution that assigns probability
\begin{equation*}
\P[q_{t+1} = q \,|\, q_t] = \frac{\sum_{\sigma} \one_{[q \in
\tau(q_t,\sigma)]}}{\sum_{\sigma} |\tau(q_t,\sigma)|}
\end{equation*}
to each state $q \in Q$.
Note that this corresponds to choosing one of the transitions from $q$ uniformly
at random.
In the case of a DFA this is equivalent to choosing $\sigma \in
\Sigma$ uniformly at random and letting $q_{t+1} = \tau(q_t,\sigma)$.

In order to study the evolution of this random walk it is useful to look at the
state distribution of the associated Markov chain.
Identifying $Q$ with $[n]$, we define the \emph{transition matrix} $\mat{P} \in
\R^{n \times n}$ of the Markov chain associated with $A$ as $\mat{P}(i,j) =
\P[q_{t+1} = j \,|\, q_t = i]$.
Note that $\mat{P}$ is a row stochastic matrix.
A \emph{distribution over states} is a vector $\mat{p} \in \R^n$ such that:
$\mat{p}(i) \geq 0$ and $\sum_{i} \mat{p}(i) = 1$.
A distribution $\mat{p}$ is \emph{stationary} with respect to the Markov chain
given by $\mat{P}$ if $\mat{p} \mat{P} = \mat{p}$.
If the distribution over states in the Markov chain at time $t$ is given by
$\mat{p}_t$, the distribution at time $t+1$ can be computed as $\mat{p}_{t+1} =
\mat{p}_t \mat{P}$.
Thus, given an initial state distribution $\mat{p}_0$, the state distribution
after $t$ steps can be computed as $\mat{p}_t = \mat{p}_0 \mat{P}^t$.
Note that in the case of a random walk on an NFA the intial distribution
corresponds to the indicator vector $\mat{e}$ such that $\mat{e}(q_0) = 1$ and
$\mat{e}(q) = 0$ for $q \in Q \setminus \{q_0\}$.
A Markov chain is said to be \emph{ergodic} if there exists a stationary
distribution $\mat{p}$ such that for every initial distribution $\mat{p}'$ one has
$\lim_{t \rightarrow \infty} \mat{p}' \mat{P}^t = \mat{p}$.

It is well-known that several properties of Markov chains can be characterized
in terms of the structure of a directed graph obtained by considering all
transitions that occur with positive probability \citep{mc-dg}.
In the case of the Markov chain associated with a random walk on a NFA, this
directed graph is the one corresponding to the transition structure given by
$\tau$: it has $n$ nodes and contains an arc from $q$ to $q'$ if and only if there
exists $\sigma \in \Sigma$ such that $q' \in \tau(q,\sigma)$.
Using this point of view, the ergodicity of the Markov chain associated with a
random walk on an NFA can be characterized in terms of the structure of its
closed communicating classes.

\begin{theorem}[\citep{mc-dg}]
If an NFA $A$ contains a unique closed communicating class $Q'$, and $Q'$ is
aperiodic, then the Markov chain associated with the random walk on $A$ is
ergodic.
\end{theorem}

\subsection{Random DFA}\label{sec:randomdfa}

In order to study the properties of typical DFA we need to define a random
process for obtaining automata sampled from the uniform distribution over the
class of all DFA with a given number of states over a fixed alphabet.
Let $\Sigma$ be an alphabet of size $r$ and $Q$ a set of $n$ states.
We construct a random DFA over $\Sigma$ and $Q$ as follows.
The initial state $q_0$ is chosen uniformly at random from $Q$.
For any $q \in Q$ and $\sigma \in \Sigma$ we determine the endpoint of
transition $\tau(q,\sigma)$ by drawing a state uniformly at random from $Q$.
Finally, for every $q \in Q$ we assign a value to $\phi(q)$ chosen uniformly at
random from $\{0,1\}$.
All these random choices are mutually independent.
Hereafter we refer to the outcome of this process as a \emph{random DFA}.

Investigating the structure of random DFA entails identifying properties that
hold with high probability with respect to this sampling process.
In particular, for a fixed alphabet size, we look for properties that occur
almost surely as the number of states in the random DFA grows; that is,
properties that hold with probability $1 - o(1)$ when the number of states $n
\rightarrow \infty$.
Noticeably enough, just a few results of this type can be found in the literature.
The first examples we are aware of appear in an early book on automata synthesis
\citep{trakbarz}.
Since then, just a few more formal results of this type have been proven, most
of them motivated by either learning theory
\citep{angluin2009learning,angluinlower}, \v{C}ern{\'y}'s conjecture
\citep{skvortsov2010synchronizing}, average behavior of DFA minimizations
algorithms \citep{bassino2012average,david2012average,de2013brzozowski}, or by
pure mathematical interest in the structure of random DFA \citep{grusho,aryeh}.
%
Among these, Grusho's result was the first to establish an interesting fact
about closed communicating classes in random DFA: with high probability they are
unique and large.

\begin{theorem}[\citep{grusho}]\label{thm:grusho}
When $n \rightarrow \infty$, a random DFA satisfies the following with
probability $1 - o(1)$:
\begin{enumerate}
\item the DFA contains a single closed communicating class,
\item the size $M$ of this closed communicating class satisfies $|M - c n| \leq
f(n)$ for some function $f(n) = o(n)$ and some constant $c$,
\item the constant $c$ above is the positive solution of $c = 1 - e^{-c r}$.
\end{enumerate}
\end{theorem}

Table~\ref{tbl:c} displays the approximate value of $c$ for several alphabet
sizes. We see that already for small alphabet sizes the communicating class
contains almost all states.
This result was established by Grusho in the form of a central limit theorem.
He also established a similar result for the number of reachable states in a
random DFA.
Using different techniques, a concentration inequality equivalent to Grusho's
result on the number of reachable states was proved in \citep{aryeh}.
This paper also shows that with high probability the number of reachable states
in a random DFA is almost the same after miminizing the automaton.
Our results provide additional information about the structure of closed
communicating classes in random DFA and their minimized versions.
In particular, we show that Grusho's closed communicating class is aperiodic,
and that minimizing a DFA with a unique closed and aperiodic communicating class
yields a DFA with a unique closed and aperiodic communicating class.

\begin{table}
\begin{center}
\begin{tabular}{c|c|c|c|c|c|c}
$r$ & 2 & 3 & 4 & 5 & 6 & 7 \\
\hline
$c$ & 0.796 & 0.940 & 0.980 & 0.993 & 0.997 & 0.999
\end{tabular}
\end{center}
\caption{Values of $c(r)$ from Theorem~\ref{thm:grusho}, truncated to the third
digit}\label{tbl:c}
\end{table}


\section{Random Walks on Random DFA are Ergodic}\label{sec:aperiodic}

In this section we state and prove our first result.
It basically states that the closed communicating class identified in Grusho's
theorem is aperiodic.
As described in Section~\ref{sec:randomwalk}, a direct consequence of this
result is that, with high probability, random walks on random DFA induce ergodic
Markov chains.

\begin{theorem}\label{thm:main}
When $n \rightarrow \infty$, with probability $1 - o(1)$ a random DFA has a
single closed communicating class which is aperiodic and whose size $M$
satisfies $|M - c n| = o(n)$ with $c$ as in Theorem~\ref{thm:grusho}.
\end{theorem}

The main idea of the proof is to bound the probability that a random DFA
contains a $k$-periodic closed communicating class for some $k \geq 2$.
We begin with two technical lemmas.

\begin{lemma}\label{lem:technical}
For any $s \geq 1$ and $0 \leq x \leq 1$, one has
\begin{equation*}
\frac{x^s}{(1-x)^{(1-x)/x}} \leq 1.2 \enspace.
\end{equation*}
\end{lemma}
\begin{proof}
Since $x^s \leq x$ for $s \geq 1$ and $0 \leq x \leq 1$, it is enough to
consider the case $s = 1$ given by $f(x) = x (1-x)^{1-1/x}$.
We start by showing that $f$ is concave on $[0,1]$. A rutinary computation shows that
\begin{equation*}
f''(x) = - \frac{g(x)}{h(x)} =
- \frac{2x^3 - x^2 + (1-x) \ln^2(1-x) - 2x(1-x) \ln(1-x)}{x^3
(1-x)^{1/x}} \enspace,
\end{equation*}
where clearly $h(x) \geq 0$ for $x \in [0,1]$.
Furthermore, $g(0) = 0$ and $g'(x) = 6 x^2 + 4x \ln(1-x) + \ln^2(1-x)$. Since
for $x \geq 0$ we have $4x \ln(1-x) + \ln^2(1-x) \geq -3x^2 - 2x^3$, we
see that $g'(x) \geq 3x^2 - 2x^3 \geq 0$ for $x \in [0,1]$. Thus $f''(x) \leq 0$
in $[0,1]$ and $f$ is concave.
Now, by concavity of $f$ and monotonicity of degree one polynomials, the
following holds for all $x, y \in [0,1]$:
\begin{equation*}
f(x) \leq \max\{f(y) + f'(y) (1-y), f(y) + f'(y) (0-y)\} \enspace.
\end{equation*}
Taking $y = 0.795$ we get $f(x) \leq 1.2$. \qed
\end{proof}

\begin{lemma}\label{lem:combinatorial}
There exists a positive constant $C$ such for any $2 \leq k \leq
m \leq n$ and $r \geq 2$ the following holds:
\begin{equation*}
\dbinom{n}{m} \dbinom{m-1}{k-1} \frac{m!}{\Gamma(m/k)^k}
\left(\frac{m}{k n}\right)^{m r} \leq
C \cdot \min\left\{m^k,2^{m}\right\} \cdot
\left( \frac{1.2}{k^{r-1}}\right)^{m} \enspace.
\end{equation*}
\end{lemma}
\begin{proof}
First note that the following bounds can be easily derived from Stirling's
approximation and common bounds for binomial coefficients:
\begin{align*}
\dbinom{n}{m} &\leq \frac{C}{m!} \frac{n^m}{e^m}
\left(\frac{n}{n-m}\right)^{n-m} \enspace, \\
\dbinom{m-1}{k-1} &\leq \min\{m^k,2^m\} \enspace, \\
\frac{1}{\Gamma(m/k)^k} &\leq
\left(\frac{e}{\frac{m}{k} + 1}\right)^m
\leq \left(\frac{e k}{m}\right)^m \enspace,
\end{align*}
where $C$ is a positive constant.
Combining these bounds in the obvious way one obtains:
\begin{align*}
\dbinom{n}{m} &\dbinom{m-1}{k-1} \frac{m!}{\Gamma(m/k)^k} \left(\frac{m}{k
n}\right)^{m r} \\
&\leq C \cdot \min\{m^k,2^m\} \cdot
\left(\left(\frac{m}{kn}\right)^{r-1}
\left(\frac{n}{n-m}\right)^{(n-m)/m} \right)^m \enspace.
\end{align*}
Finally, invoking Lemma~\ref{lem:technical} with $x = m/n$ we get:
\begin{equation*}
\left(\frac{m}{kn}\right)^{r-1}
\left(\frac{n}{n-m}\right)^{(n-m)/m}
\leq
\frac{1.2}{k^{r-1}} \enspace.
\end{equation*}
\qed
\end{proof}

Now let $m \leq n$ and $2 \leq k \leq m$.
A DFA contains a $k$-periodic closed communicating class of size $m$ if and only
if there exists a subset of states $Q' \subseteq Q$ with $|Q'| = m$ that can be
partitioned into $k$ disjoint subsets $(Q'_0, \ldots, Q'_{k-1})$ such that
$\tau_1(Q_i) = Q_{i+1 \mod k}$ for all $0 \leq i \leq k-1$.
We use $E_{m,k}$ to denote the event that a random DFA contains a $k$-periodic
closed communicating class of size $m$.
The following lemma bounds the probability of $E_{m,k}$.

\begin{lemma}
There exists a constant $\alpha > 0$ such that for any $m \leq n$ and any $2
\leq k \leq m$ one has $\P[E_{m,k}] = O(e^{-\alpha m})$.
\end{lemma}
\begin{proof}
Let $Q' \subseteq Q$ be a fixed subset with $m$ states and
$(Q'_0,\ldots,Q'_{k-1})$ a fixed partition of $Q'$ into $k$ parts.
Let us write $m_i = |Q'_i|$.
When assigning the transitions of a random DFA, the probability that $Q'$
is a $k$-periodic closed communicating class with this particular
partition is at most
\begin{equation*}
\left(\frac{m_1}{n}\right)^{m_0 r}
\left(\frac{m_2}{n}\right)^{m_1 r} \cdots
\left(\frac{m_0}{n}\right)^{m_{k-1} r} =
\left(\frac{m_1^{m_0} m_2^{m_1} \cdots m_0^{m_{k-1}}}{n^m}\right)^{r}
\enspace.
\end{equation*}
Note that the function $f(x_0,\ldots,x_{k-1}) = x_1^{x_0} \cdots
x_{k-1}^{x_{k-2}} x_0^{x_{k-1}}$ under the constraints $x_i > 0$ and $x_0 +
\cdots + x_{k-1} = m$ is maximized for $x_i = m/k$.

To count the number of partitions of a set of $m$ states into an \emph{ordered}
tuple of $k$ sets of states, imagine that we first choose the sizes
$(m_0,\ldots,m_{k-1})$ such that $m_i > 0$ and $m_0 + \cdots + m_{k-1} = m$, and
then we choose each $Q'_i$ of size $m_i$. Let $s(m,k)$ denote the number of
tuples of sizes $(m_0,\ldots,m_{k-1})$ satisfying the conditions. Using $s(m,1)
= 1$, $s(m,m) = 1$, and $s(m,k) = \sum_{j = 1}^{n-(k-1)} s(m-j,k-1)$, it is easy
to show that $s(m,k) \leq \binom{m-1}{k-1}$. Furthermore, once the sizes are
chosen, the number of ways in which the sets in the partition can be chosen is
given by the multinomial coefficient
\begin{equation*}
\dbinom{m}{m_0,\ldots,m_{k-1}} = \frac{m!}{\Gamma(m_0+1) \cdots
\Gamma(m_{k-1}+1)} \enspace,
\end{equation*}
which is maximized by the (non-necessarily integer) choice $m_i = m/k$.
Combining the above observations we get
\begin{equation*}
\P[E_{m,k}] \leq
\dbinom{n}{m} \dbinom{m-1}{k-1} \frac{m!}{\Gamma(m/k)^k}
\left(\frac{m}{k n}\right)^{m r}
\enspace,
\end{equation*}
which by Lemma~\ref{lem:combinatorial} implies that
\begin{equation*}
\P[E_{m,k}] \leq C \cdot \min\left\{m^k,2^{m}\right\} \cdot
\left( \frac{1.2}{k^{r-1}}\right)^{m}
\enspace.
\end{equation*}
Now note that since $r \geq 2$, for $k = 2$ we have
$\P[E_{m,2}] \leq C \cdot m^2 \cdot 0.6^m$, and for $k \geq 3$ we have
$\P[E_{m,k}] \leq C \cdot 0.8^m$.
Therefore we can conclude that for any $2 \leq k \leq m$ one has $\P[E_{m,k}] =
O(e^{-\alpha m})$ for some $\alpha > 0$. \qed
\end{proof}

Now we can use this lemma to give a proof for Theorem~\ref{thm:main} using a
union bound argument.

\begin{proof}[of Theorem~\ref{thm:main}]
Let $E$ denote the event that a random DFA has a periodic closed communicating
class and $E_m$ the event that a random DFA has a periodic closed communicating
class of size $m$.
Let $[n] = U \cup L$ denote a partition of $[n]$ into
the sets of unlikely and likely sizes of a closed communicating class of a
random DFA. According to Theorem~\ref{thm:grusho} we can take $L = [cn - f(n),
cn + f(n)]$ and $U = [1,c n-f(n)) \cup (c n+f(n),n]$ where $f(n) = o(n)$.
Note that we have $|L| = 2 f(n) = o(n)$ and every $m \in L$ satisfies
$m = n (c+o(1))$.
Furthermore, if $E_U$ denotes the event that a random DFA contains a closed
communicating class whose size belongs to $U$, by Theorem~\ref{thm:grusho} we
have $\P[E_U] = o(1)$.

Using these facts we can now conclude that the probability that a random DFA has
a periodic closed communicating class is
\begin{align*}
\P[E] &\leq \P[E_U] + \sum_{m \in L} \P[E_m]
\leq o(1) + \sum_{m \in L} \sum_{k = 2}^m \P[E_{m,k}] \\ 
&\leq o(1) + o(n) \cdot n (c + o(1)) \cdot O(e^{-\alpha n (c + o(1))}) = o(1)
\enspace.
\end{align*}
This implies that with probability $1 - o(1)$ all closed communicating classes
of a random DFA are aperiodic. Since the event in Theorem~\ref{thm:grusho}
together with the event that all closed communicating classes are aperiodic
imply the event in Theorem~\ref{thm:main}, we conclude that this last event
holds with probability $1 - o(1)$. \qed
\end{proof}


\section{Effect of DFA Minimization on Aperiodic Closed Communicating
Classes}\label{sec:minimization}

Now we present our second result which studies the effect of DFA minimization on
the structure of aperiodic closed communicating classes.
In particular, we show that minimizing a DFA with a single closed communicating
which is also aperiodic yields a DFA with that same property.
When combined with Theorem~\ref{thm:main} this implies that ergodicity of random
walks holds even if one considers minimized versions of randomly generated DFA.

\begin{theorem}\label{thm:main2}
Let $A$ be a DFA and $\tilde{A}$ the output of the algorithm in
Figure~\ref{fig:algmin} on input $A$.
If $A$ contains a single closed and aperiodic communicating class, then
$\tilde{A}$ also contains a single closed and aperiodic communicating class.
\end{theorem}

We begin with a simple lemma about merges of NFA containing a single closed
communicating class.

\begin{lemma}\label{lem:psirecurrent}
Let $A = \langle \Sigma,Q,q_0,\tau,\phi\rangle$ be an NFA, $\Psi$ a merge
operation, and $\tilde{A} = \Psi(A) = \langle
\Sigma,\tilde{Q},\tilde{q}_0,\tilde{\tau},\tilde{\phi}\rangle$.
Let $\Qr$ and $\tQr$ denote the sets of recurrent states in $A$ and $\tilde{A}$
respectively.
If $\Qr$ contains a single closed communicating class, then $\Psi(\Qr) \subseteq
\tQr$.
\end{lemma}
\begin{proof}
Let $q \in \Qr$ be recurrent in $A$. To show that $\Psi(q)$ is recurrent in
$\tilde{A}$ we must show that we have $\Psi(q) \in \tau_\star(\tilde{q})$ for
every $\tilde{q} \in \tau_\star(\Psi(q))$.
Let $q'$ be an arbitrary state in $\Psi^{-1}(\tilde{q})$.
Then, since $\Qr$ forms a single communicating class, by
Fact~\ref{fact:recurrent} we have $q \in \tau_\star(q')$, which yields $\Psi(q)
\in \tilde{\tau}_\star(\tilde{q})$. \qed
\end{proof}

The next lemma shows that having no unreachable states is a property of NFA
conserved by merge operations.

\begin{lemma}\label{lem:min:reachable}
Let $A = \langle \Sigma,Q,q_0,\tau,\phi\rangle$ be an NFA
and $\Psi$ a merge operation. If $A$ has no unreachable states, then $\Psi(A)$
has no unreachable states.
\end{lemma}
\begin{proof}
Let $\tilde{A} = \Psi(A) = \langle
\Sigma,\tilde{Q},\tilde{q}_0,\tilde{\tau},\tilde{\phi}\rangle$.
Let $\tilde{q} \in \tilde{Q}$ and choose some $q \in \Psi^{-1}(\tilde{q})$.
By hypothesis we have $q \in \tau_\star(q_0)$, which implies $\tilde{q} \in
\tau_\star(\tilde{q}_0)$. \qed
\end{proof}

Now we are ready to prove the first half of Theorem~\ref{thm:main2}: that
having a single closed communicating class is a property invariant under merge
operations.

\begin{lemma}\label{lem:min:ergodic}
Let $A = \langle \Sigma,Q,q_0,\tau,\phi\rangle$ be an NFA with no unaccessible
states and $\Psi$ a merge operation.
If $A$ contains a single closed communicating class, then $\Psi(A)$ also
contains a single closed communicating class.
\end{lemma}
\begin{proof}
Let $\tilde{A} = \Psi(A) = \langle
\Sigma,\tilde{Q},\tilde{q}_0,\tilde{\tau},\tilde{\phi}\rangle$ and write
$\tilde{Q} = \tQt \cup \tQr$ for the partition of $\tilde{Q}$ into transient and
recurrent states.
Since every closed communicating class will be contained in $\tQr$,
it suffices to show that for every pair of states $\tilde{q}, \tilde{q}' \in
\tQr$ we have $\tilde{q}' \in \tilde{\tau}_\star(\tilde{q})$.
Start by choosing arbitrary states $q \in \Psi^{-1}(\tilde{q})$ and $q' \in
\Psi^{-1}(\tilde{q}')$.
Because $A$ has a single closed communicating class and no unaccessible
states, we must have $\Qr \subseteq \tau_\star(q) \cap
\tau_\star(q')$ by Fact~\ref{fact:recurrent}.
Now choose an arbitrary $q'' \in \Qr$ and note that by
Lemma~\ref{lem:psirecurrent} we must have $\tilde{q}'' = \Psi(q'') \in \tQr$.
Finally, to build a path from $\tilde{q}$ to $\tilde{q}'$ we observe the
following: $\tilde{q}'' \in \tilde{\tau}_\star(\tilde{q})$ because $q'' \in
\tau_\star(q)$, and $\tilde{q}' \in \tau_\star(\tilde{q}'')$ because
$\tilde{q}'' \in \tau_\star(\tilde{q}')$ and both are recurrent. \qed
\end{proof}

The last ingredient is given by the following lemma which states that
aperiodicity is also invariant under merge operations.

\begin{lemma}\label{lem:min:aperiodic}
Let $A = \langle \Sigma,Q,q_0,\tau,\phi\rangle$ be an NFA with no unreachable
states containing a single closed communicating class and $\Psi$ a merge
operation.
If the closed communicating class in $A$ is aperiodic, then $\Psi(A)$ contains
a single closed communicating class which is also aperiodic.
\end{lemma}
\begin{proof}
Let $\tilde{A} = \Psi(A) = \langle
\Sigma,\tilde{Q},\tilde{q}_0,\tilde{\tau},\tilde{\phi}\rangle$ and write
$\tilde{Q} = \tQt \cup \tQr$ for the partition of $\tilde{Q}$ into transient and
recurrent states.
By Lemma \ref{lem:min:ergodic} we know that $\tQr$ contains a single closed
communicating class.
Suppose that $\tQr$ is $k$-periodic for some $k > 1$.
This means there exists a partition $\tilde{Q}_0, \ldots,
\tilde{Q}_{k-1}$ of $\tQr$ such that
$\tilde{\tau}_1(\tilde{Q}_i) = \tilde{Q}_{i+1}$ for $0 \leq i \leq k-1$.%
\footnote{All subindex calculations throughout this proof are
performed modulo $k$.}
We claim that then the sets $Q_i = \Psi^{-1}(\tilde{Q}_i) \cap \Qr$ induce
a $k$-periodic partition of the closed communicating class $\Qr$ of
$A$.
Note that the sets $Q_i$ are disjoint by construction.
In addition, by Lemma~\ref{lem:psirecurrent} we necessarily have $\Qr = Q_0 \cup
\cdots \cup Q_{k-1}$.
To prove that this partition is $k$-periodic, we will show that, for any $0 \leq
i \leq k-1$, the
two inclusions in $\tau_1(Q_i) = Q_{i+1}$ hold.
First note that since $Q_i \subseteq \Qr$, by construction we have
$\tau_1(Q_i) \subseteq \Qr$.
Now suppose that for some $q \in Q_i$ and $\sigma \in \Sigma$ there exists
a state $q' \in \tau(q,\sigma) \cap Q_j$ with $j \neq i+1$.
Then we have $\Psi(q') \in \tilde{Q}_j$ and $\Psi(q') \in
\tilde{\tau}(\Psi(q),\sigma) \subseteq \tilde{Q}_{i+1}$, which is impossible by
the choice of $j$ and the assumption that $\tQr$ is $k$-periodic.
Hence, necessarily $\tau_1(Q_i) \subseteq Q_{i+1}$.
Now suppose there exists $q \in Q_{i+1} \setminus \tau_1(Q_i)$. Then,
since $q$ is recurrent, there must exist $\sigma \in \Sigma$ and $q' \in Q_j$
with $j \neq i$ such that $q \in \tau(q',\sigma)$. In this case, we have
$\Psi(q) \in \tilde{Q}_{i+1}$ and $\Psi(q) \in
\tilde{\tau}(\Psi(q'),\sigma) \subseteq \tilde{Q}_{j+1}$, which again is
impossible. Thus, we get $Q_{i+1} \subseteq \tau_1(Q_i)$. \qed
\end{proof}

Theorem~\ref{thm:main2} now follows immediately from
Lemma~\ref{lem:min:aperiodic} because $\tilde{A}$ is obtained from $A$ by
removing all unreachable states and applying a sequence of merge operations.

\section{Conclusion}\label{sec:conclusion}

We have shown that random walks on typical DFA will converge to a unique
stationary distribution by studying the structure of closed communicating
classes in randomly constructed DFA.
However, our results do not provide any bounds as to the speed at which this
convergence takes place.
Usual methods for establishing rapid mixing of Markov chains do not apply in our
case because in general a random walk on a DFA is neither lazy nor reversible.
As future work we plan to investigate wether a more precise study of the
structure of closed communicating classes in random DFA can be used to bound the
mixing speed for these random walks.

\subsubsection*{Acknowledgements.}
The author wants to thank R.\ Gavald{\`a} and D.\ Angluin for helpful
discussions, and D.\ Berend and A.\ Kontorovich for their comments on an early
version of this manuscript.
Most of this work was done while the author was a PhD student at Universitat
Polit{\`e}cnica de Catalunya supported by project BASMATI
(TIN2011-27479-C04-03).

\bibliographystyle{splncs}
\bibliography{paper}

\end{document}